\documentclass[conference]{IEEEtran}
\IEEEoverridecommandlockouts

\usepackage{cite}
\usepackage{amsmath,amssymb,amsfonts}
\usepackage{algorithm}
\usepackage[noend]{algpseudocode}
\usepackage{graphicx}
\usepackage{textcomp}
\usepackage{xcolor}
\usepackage{array}
\usepackage{booktabs}
\usepackage{multirow}
\usepackage{comment}
\usepackage{caption}
\usepackage{subcaption}
\usepackage{tabularx}
\usepackage{bm}

\usepackage{tikz}
\usepackage{pgfplots}
\pgfplotsset{scaled y ticks=false}
\usepackage{mathrsfs}
\usepackage{mathtools}
\usetikzlibrary{arrows}
\usepgfplotslibrary{external}
\usepackage{xcolor}
\usepackage{pifont}
\pgfplotsset{grid style={dashed,gray}}
\pgfplotsset{minor grid style={dotted,gray}}
\pgfplotsset{major grid style={dashed,gray}}

\makeatletter
\newcommand{\multiline}[1]{%
  \begin{tabularx}{\dimexpr\linewidth-\ALG@thistlm}[t]{@{}X@{}}
    #1
  \end{tabularx}
}
\makeatother

\newcommand*{\rom}[1]{\expandafter\@slowromancap\romannumeral #1@}
\makeatother
\allowdisplaybreaks

\usepackage{setspace}
\usepackage{mathtools}
\makeatother
\DeclareMathOperator*{\argmax}{argmax} 
\makeatother
\DeclareMathOperator*{\argmin}{argmin} 

\usepackage[T1]{fontenc}
\usepackage{float}

\graphicspath{ {Images/} }

\def\BibTeX{{\rm B\kern-.05em{\sc i\kern-.025em b}\kern-.08em
    T\kern-.1667em\lower.7ex\hbox{E}\kern-.125emX}}

\usepackage{acronym}
\newacro{nf} [NF] {near-field}
\newacro{ff} [FF] {far-field}
\newacro{5g} [5G] {fifth-generation}
\newacro{6g} [6G] {sixth-generation}
\newacro{isac} [ISAC] {integrated sensing and communication}
\newacro{ul} [UL] {uplink}
\newacro{mu} [MU] {multi-user}
\newacro{csi} [CSI] {channel state information}

\newacro{em} [EM] {electromagnetic}
\newacro{siso} [SISO] {single input single output}
\newacro{miso} [MISO] {multiple input single output}
\newacro{mimo} [MIMO] {multiple input multiple output}
\newacro{xl-mimo} [XL-MISO] {extremely large-scale MIMO}
\newacro{ras} [RAs] {reconfigurable antennas}
\newacro{los} [LoS] {line-of-sight}
\newacro{nlos} [NLoS] {non-line-of-sight}
\newacro{shod} [SHOD] {spherical harmonious orthogonal decomposition}
\newacro{ofdm} [OFDM] {orthogonal frequency division multiplexing}
\newacro{dof} [DoF] {degree of freedom}
\newacro{fim} [FIM] {Fisher information matrix}
\newacro{em} [EM] {electromagnetics}
\newacro{er-fas} [ER-FAS] {electromagnetically reconfigurable FAS}
\newacro{eras} [ERAs] {electromagnetically reconfigurable antennas}
\newacro{aod} [AOD] {angle-of-departure}
\newacro{aoa} [AOA] {angle-of-arrival}
\newacro{aoas} [AOAs] {angles-of-arrival}
\newacro{sp} [SP] {scatter point}
\newacro{ml} [ML] {maximum likelihood}
\newacro{mse} [MSE] {mean square error}
\newacro{snr} [SNR] {signal-to-noise ratio}
\newacro{lmr} [LMR] {line-of-sight to multipath ratio}

\newacro{sv} [SV] {Saleh-Valenzuela}

\newacro{rmse} [RMSE] {root mean square error}
\newacro{crb} [CRB] {Cram\'er-Rao bound}
\newacro{peb} [PEB] {position error bound}
\newacro{kld} [KLD] {Kullback–Leibler divergence}
\newacro{siso} [SISO] {single-input-single-output}
\newacro{mimo} [MIMO] {multiple-input multiple-output}
\newacro{mcrb} [MCRB] {misspecified Cram\'er-Rao bound}
\newacro{bs} [BS] {base station}
\newacro{ue} [UE] {user equipment}
\newacro{arv} [ARV] {array response vector}
\newacro{earv} [EARV] {effective ARV}
\newacro{upa} [UPA] {uniform planar array}
\newacro{rf} [RF] {radio frequency}
\newacro{bb} [BB] {baseband}
\newacro{ris} [RIS] {reconfigurable intelligent surface}
\newacro{rms} [RMS] {root mean square}
\newacro{psd} [PSD] {positive semidefinite}
\newacro{lmi} [LMI] {linear matrix inequality}
\newacro{bca} [BCA] {block-coordinate ascent}
\newacro{bcd} [BCD] {block-coordinate descent}

\newacro{ma} [MA] {movable antenna}
\newacro{6dma} [6DMA] {six-dimensional movable antenna}
\newacro{pra} [PRA] {pattern reconfigurable antenna}
\newacro{ra} [RA] {reconfigurable antenna}
\newacro{pra_p} [PRAs] {pattern reconfigurable antennas}
\newacro{espar} [ESPAR] {electronically steerable parasitic array radiator}
\newacro{music} [MUSIC] {multiple signal classification}
\newacro{rss} [RSS] {received signal strength}
\newacro{toa} [TOA] {time of arrival}
\newacro{wsn} [WSN] {wireless sensor network}
\newacro{nm} [NM] {Nelder-Mead}
\newacro{ls} [LS] {least-squares}
\newacro{sdp} [SDP] {semidefinite program}
\newacro{iot} [IoT] {internet of things}

\newacro{mpc} [MPC] {multipath component}
\newacro{mpc_p} [MPCs] {multipath components}

\newacro{fas} [FAS] {fluid antenna system}
\newacro{fas_p} [FAS] {fluid antenna systems}
\newacro{sr-fas} [SR-FAS] {spatially reconfigurable FAS}
\newacro{ris} [RIS] {reconfigurable intelligent surface}
\newacro{ris_p} [RISs] {reconfigurable intelligent surfaces}
\newacro{rcs} [RCS] {radar cross section}

\usepackage{titlesec}

\usepackage{accents}

\usepackage{amsthm}

\theoremstyle{plain}

\newtheoremstyle{iremark}
  {\topsep}   
  {\topsep}   
  {\upshape}  
  {0.2in}       
  {\itshape}  
  {.}         
  {5pt plus 1pt minus 1pt} 
  {\thmname{#1}\thmnumber{ \itshape#2}\thmnote{ (#3)}} 

\usepackage{amsthm}

\newtheorem{theorem}{Theorem}
\newtheorem{lemma}[theorem]{Lemma}

\newtheorem{proposition}{Proposition}

\theoremstyle{definition}
\newtheorem*{proof}{Proof}

\makeatletter
\newcommand*\rel@kern[1]{\kern#1\dimexpr\macc@kerna}
\newcommand*\widebar[1]{%
  \begingroup
  \def\mathaccent##1##2{%
    \rel@kern{0.8}%
    \overline{\rel@kern{-0.8}\macc@nucleus\rel@kern{0.2}}%
    \rel@kern{-0.2}%
  }%
  \macc@depth\@ne
  \let\math@bgroup\@empty \let\math@egroup\macc@set@skewchar
  \mathsurround\z@ \frozen@everymath{\mathgroup\macc@group\relax}%
  \macc@set@skewchar\relax
  \let\mathaccentV\macc@nested@a
  \macc@nested@a\relax111{#1}%
  \endgroup
}
\makeatother

\usepackage{fancyhdr}

\fancypagestyle{arxivnotice}{
  \fancyhf{} 
  
  \fancyhead[C]{\small \textit{This paper has been accepted for publication in the 2026 IEEE International Conference on Communications (ICC).}}
  
  \fancyfoot[C]{\scriptsize \copyright~2026 IEEE. Personal use of this material is permitted. Permission from IEEE must be obtained for all other uses, in any current or future media, including reprinting/republishing this material for advertising or promotional purposes, creating new collective works, for resale or redistribution to servers or lists, or reuse of any copyrighted component of this work in other works.}
}

\begin{document}
\title{
Near-Field Localization via Reconfigurable Antennas
}

\author{
\IEEEauthorblockN{Alireza Fadakar\IEEEauthorrefmark{1}, 
Yuchen Zhang\IEEEauthorrefmark{2},
Hui Chen\IEEEauthorrefmark{3},
Musa Furkan Keskin\IEEEauthorrefmark{3}, Henk Wymeersch\IEEEauthorrefmark{3}, Andreas F. Molisch\IEEEauthorrefmark{1}}
\IEEEauthorblockA{\IEEEauthorrefmark{1}University of Southern California, Los Angeles, CA, USA
\\\{fadakarg, molisch\}@usc.edu}
\IEEEauthorblockA{\IEEEauthorrefmark{2}King Abdullah University of Science and Technology (KAUST), Thuwal, Kingdom of Saudi Arabia
\\ yuchen.zhang@kaust.edu.sa}
\IEEEauthorblockA{\IEEEauthorrefmark{3}Chalmers University of Technology, Gothenburg, Sweden
\\\{hui.chen, furkan, henkw\}@chalmers.se}
}

\maketitle
\thispagestyle{arxivnotice}

\begin{abstract}
Reconfigurable antennas (RAs) utilize the electromagnetic (EM) domain to provide dynamic control over antenna radiation patterns, which offers an effective way to enhance power efficiency in wireless links. 
Unlike conventional arrays with fixed element patterns, RAs enable on-demand beam-pattern synthesis by directly controlling each antenna's EM characteristics.
While existing research on RAs has primarily focused on improving spectral efficiency, this paper explores their application for downlink localization. 
Moreover, the majority of existing works focus on far-field scenarios with little attention on near-field (NF). 
Motivated by these gaps, we consider a synthesis model in which each antenna generates desired beampatterns from a finite set of EM basis functions. 
We then formulate a joint optimization problem for the baseband (BB) and EM precoders with the objective of minimizing the user equipment (UE) position error bound (PEB) in NF conditions. 
Our analytical derivations and extensive simulation results demonstrate that the proposed hybrid precoder design for RAs significantly improves UE positioning accuracy compared to traditional non-reconfigurable arrays.
\end{abstract}

\bstctlcite{IEEEexample:BSTcontrol}

\section{Introduction}
Localization has become widespread in applications such as wireless networks, augmented reality, and the \ac{iot} \cite{del2018survey, fadakar2024multi} within the framework of \ac{isac}. 
Localization methods are commonly categorized by whether the source lies in the receiver's \ac{ff} or \ac{nf} region \cite{Friedlander2019NF}. 
Sources in the \ac{nf} region produce spherical wavefronts at the receiver, whereas \ac{ff} sources are seen as plane waves and are thus described only by \ac{aoa} or \ac{aod} \cite{fadakar2025near,Ozturk2024pixel}.
Therefore, algorithms developed under the \ac{ff} assumption are generally unsuitable for \ac{nf} source localization. 

Both \ac{ff} and \ac{nf} localization have been extensively investigated across diverse scenarios in the literature (e.g., \cite{Friedlander2019NF, Lu2024NF, Yang2025NF, fascista2022ris, fadakar2024multi, fadakar2025near, fadakar2025stacked} and the references therein). 
Most prior works have emphasized beamforming via per-element phase and amplitude control to steer power toward desired directions. 
A major limitation in this body of work, however, is the common practice of modeling antenna arrays using elements with fixed, frequently identical radiation patterns \cite{liu2025tri}.
The constraints imposed by conventional pattern-fixed antennas are becoming increasingly relevant in emerging applications, such as \ac{nf} communications \cite{liu2025tri,You2025jsac}.
In particular, very narrow element patterns yield high gain only within limited angular sectors and suffer steep loss near endfire, while broad element patterns improve coverage but reduce peak gain. 

\Ac{ra} designs address this trade-off by altering internal settings to reshape radiation patterns\cite{Zhang2022PRA, Fadakar2026Hybrid}, tune frequency response\cite{Song2014frequency}, or change polarization\cite{Zheng2025Polarization}. 
Among these capabilities, dynamic control of the element radiation pattern is particularly impactful for communication and sensing tasks because it directly affects both transmitted and received power and thus overall system performance \cite{liu2025tri, Fadakar2026Hybrid}. 
In this paper, we also focus on radiation pattern-based \ac{ra}.
The element radiation of \ac{ra} can be reconfigured in real time by modifying the current distribution of an \ac{rf} radiator \cite{wang2025RA}. 
Typical fabrication techniques comprise pixelated parasitic layouts \cite{Zhang2025pixel, liu2025tri}, \ac{espar} architectures featuring a single driven element surrounded by passive parasitics \cite{Han2021Characteristic, liu2025tri}, and liquid-metal fluidic implementations \cite{wang2025RA, Fadakar2026Hybrid}.
Independent of the underlying fabrication process, an \ac{ra} can actively modify its \ac{em} characteristics. 
This flexibility adds new \ac{dof} in the \ac{em} domain, which can be exploited to alleviate conventional performance bottlenecks and improve spectral and spatial efficiency of wireless communication and sensing systems.

Recent advances in \ac{ra} technology have shown growing interest in their integration within 
antenna arrays \cite{Ying2024Reconfigurable}. 
While several works have demonstrated the spectral-efficiency benefits of \ac{ra} arrays\cite{liu2025tri, wang2025RA}, their potential for wireless sensing, particularly for localization, has received limited attention. 
Furthermore, the majority of prior studies in both communications and sensing assume \ac{ff} conditions \cite{Fadakar2026Hybrid}. Motivated by these gaps, this paper targets localization in \ac{nf} scenarios. 
In particular, the main contributions of this work are summarized as follows:
\begin{itemize}
\item 
We develop a synthesis-based signal model for localization of a \ac{nf} user with an \ac{ra} array, where each element forms a specific beampattern from a finite set of basis functions. 
From this model, we derive expressions for the \ac{fim} and the resulting \ac{crb}, and use these expressions to establish theoretical limits on \ac{nf} localization performance. 
\item
Building on this analysis, we formulate a joint design problem that jointly optimizes the \ac{em} and digital precoders to minimize the \ac{crb}.
We propose a computationally efficient, multi-beam scheme that accounts for \ac{ue} position uncertainty based on the optimal parameter optimization with \ac{ue} position knowledge.
\item 
We propose a \ac{ml}-based \ac{nf} localization estimator and validate it through extensive simulations that compare the \ac{crb} and localization \ac{rmse} against conventional fixed-pattern arrays. 
Numerical results demonstrate pronounced localization improvements offered by \ac{ra}-equipped arrays.
\end{itemize}

Matrices and vectors are represented by boldface uppercase ($\mathbf{X}$) and lowercase ($\mathbf{x}$) letters, respectively. 
The superscripts $(\cdot)^{\mathsf{T}}$, $(\cdot)^{\mathsf{H}}$, and $(\cdot)^{-1}$ denote the transpose, Hermitian (conjugate transpose), and matrix inverse, respectively. 
Standard operations include the $\ell_2$ (Euclidean) norm, denoted by $\lVert\mathbf{x}\rVert$, and the Hadamard product, $\mathbf{A}\odot\mathbf{B}$. 
We use $[\mathbf{x}_1,\ldots,\mathbf{x}_n]$ for the horizontal concatenation of vectors, $\mathrm{diag}(\mathbf{x})$ for a diagonal matrix with the elements of $\mathbf{x}$ on its main diagonal, and $\mathrm{blkdiag}\bigl\{\mathbf{A}_{1},\,\dots,\,\mathbf{A}_{n}\bigr\}$ for a block diagonal matrix. $\mathbf{I}_n$ denotes the $n\times n$ identity matrix.

\section{System and Signal Model}
\subsection{System Model}
As depicted in Fig.~\ref{fig:system-model}, we analyze a \ac{miso} wireless system where the \ac{bs} is equipped with a \ac{upa}. 
This array consists of $M=M^{h}M^{v}$ \ac{ras}, with $M^h$ and $M^v$ representing the number of rows and columns, respectively. 
The position of the center of the \ac{bs} and the $m$-th \ac{ra} are denoted by $\mathbf{p}_b$ and $\mathbf{p}_{b,m}$, respectively. 
Each \ac{ra} is capable of synthesizing a desired pattern. 
The \ac{ue} is assumed to have a single isotropic antenna located at the position $\mathbf{p}_u$. 
Assuming each \ac{mpc} arises from a single reflection at point scatterers, 
we model multipath using $I$ scatterers with positions $\{\mathbf{p}_s^{(i)}\}_{i=1}^{I}$ based on a geometric channel model, which will be detailed in the subsequent subsection. 
For localization, the \ac{bs} transmits $N_t$ narrowband, single-carrier pilots to the \ac{ue}.

\begin{figure}
\centering
\includegraphics[width=0.8\columnwidth]{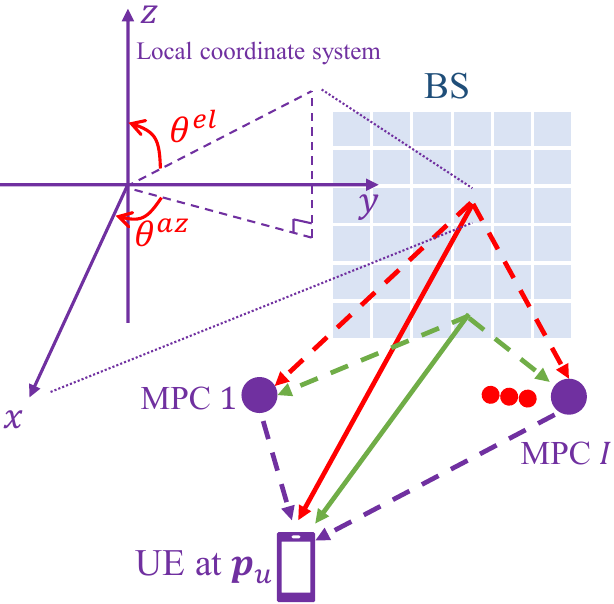}
\vspace{0.1cm}
\caption{
Considered near-field downlink localization system using an \ac{ra} array.
}
\label{fig:system-model}
\end{figure}

\subsection{Signal Model}
The received signal corresponding to the $t$-th symbol can be expressed as:
\begin{equation}\label{eq:y-def}
y_{t}
=
\sqrt{P}
\mathbf{h}_{t}^{\mathsf{T}}
\mathbf{f}_{t}
+
v_{t},
\end{equation}
where $P$ denotes the transmit power, $\mathbf{f}_{t} \in \mathbb{C}^{M}$ is the digital precoder for the $t$-th transmission, $\mathbf{h}_{t} \in \mathbb{C}^{M}$ represents the wireless channel between the \ac{bs} and the \ac{ue}, and $v_{t} \sim \mathcal{CN}(0,\sigma^2_v)$ is the additive white Gaussian noise. 
The digital precoders are normalized to satisfy $\sum_{t=1}^{N_t} \mathbf{f}_t^\mathsf{H}\mathbf{f}_t=1$, which ensures a constant total transmit power over all transmissions. 
The channel's dependence on the time index $t$ stems from the time-variant nature of the \ac{ra}, as will be elaborated upon later.

Under the assumption of a geometric channel model that includes a \ac{los} path and $I$ \ac{nlos} paths, the spatial-domain channel vector $\mathbf{h}_{t}$ is given by \cite{wang2025RA,Ying2024Reconfigurable}:
\begin{equation}\label{eq:ch-def}
\mathbf{h}_{t}
=
\sum_{i=0}^{I}
\beta_i
\mathbf{q}_{t}(\mathbf{p}^{(i)}_s),
\end{equation}
where $\beta_i$ is the overall complex channel gain of the $i$-th path, which can be modeled as \cite{Rahal2024RIS, Tang2021path, fadakar2025near}: 
\begin{align}
\label{eq:beta}
\beta_i =
\begin{cases}
\frac{\lambda}{4\pi d_{u}} 
e^{j\psi_0}
& i=0, \text{\ac{los} path,}\\
\frac{\sqrt{4\pi \sigma_i}\lambda}{16\pi^2 d_{b,i} d_{u,i}}
e^{j\psi_i}
& \text{$i$-th \ac{nlos} path,}
\end{cases}
\end{align}
where $\lambda$ is the wavelength and $\sigma_i$ is the \ac{rcs} of the $i$-th scatterer; alternative models for the attenuation of the \ac{mpc}s can also be easily incorporated. 
In \eqref{eq:beta}, $\psi_i$ is the random phase shift corresponding to the $i$-th path. 
The distances are defined as $d_{u}=\lVert \mathbf{p}_{b}-\mathbf{p}_u\rVert$ (from \ac{bs} to the \ac{ue}), $d_{b,i}=\lVert \mathbf{p}_{b}-\mathbf{p}_{s,i}\rVert$ (from \ac{bs} to the $i$-th scatterer), and $d_{u,i}=\lVert \mathbf{p}_{s}^{(i)}-\mathbf{p}_u\rVert$ (from the $i$-th scatterer to the \ac{ue}).

In \eqref{eq:ch-def}, the vector $\mathbf{q}_{t}(\mathbf{p}^{(i)}_s) \in \mathbb{C}^{M}$ is a composite term that accounts for the \ac{arv}, and the radiation patterns of the \ac{ra} array during the $t$-th transmission in the $i$-th path with $i=0$ corresponding to the \ac{los} path, and is defined as follows: 
\begin{equation}\label{eq:q-def}
\mathbf{q}_{t}(\mathbf{p}^{(i)}_s)
=
\mathbf{g}_{t}(\mathbf{p}^{(i)}_s)
\odot
\mathbf{a}(\mathbf{p}^{(i)}_s),
\end{equation}
where, the vector $\mathbf{a}(\mathbf{p})$ denotes the \ac{nf} \ac{arv} of the \ac{bs} whose $m$-th element is defined as follows \cite{Rahal2024RIS, fadakar2025near}:
\begin{equation}\label{eq:NF-str-vec}
[\bm{a}(
\bm{p}
)]_m
=
\text{exp}
\left(
-j
\frac{2\pi}{\lambda}
(
\lVert \bm{p}-\bm{p}_{b,m}\rVert
-
\lVert \bm{p}-\bm{p}_{b}\rVert
)
\right).
\end{equation}
Finally, the complex radiation patterns of the \ac{ras} are represented by the vector $\mathbf{g}_{t}(\mathbf{p}) \in \mathbb{C}^{M}$, with the $m$-th element:
\begin{equation}\label{eq:radiation_pat}
[\mathbf{g}_{t}(\mathbf{p})]_m
=
\mathbf{e}_{m,t}^{\mathsf{H}}
\mathbf{b}(\boldsymbol{\theta}_m),
\end{equation}
where $\boldsymbol{\theta}_m=[\theta^{\text{el}}_m,\theta^{\text{az}}_m]^{\mathsf{T}}$ is the 2D-\ac{aod} from the $m$-th antenna towards the position $\mathbf{p}$,  $\mathbf{b}(\boldsymbol{\theta}) = [b_{1}(\boldsymbol{\theta}), \dots, b_{Q}(\boldsymbol{\theta})]^{\mathsf{T}}$ represents a vector of $Q$ orthonormal basis functions, and the vector $\mathbf{e}_{m,t}\in\mathbb{C}^{Q}$ denotes the \ac{em} precoder vector, which assigns weights to the basis functions to synthesize the desired radiation pattern of the $m$-th \ac{ra} during the $t$-th transmission. 
The details are provided in the following subsection.

\subsection{Reconfigurability Model}
Motivated by the proven usefulness of \ac{shod} functions for plane-wave expansion, array-manifold factorization, and radiation-pattern decomposition, we construct the basis functions using \ac{shod} \cite{Ying2024Reconfigurable, Fadakar2026Hybrid}. 
These functions inherently satisfy orthonormality on the unit sphere, making them well-suited for spherical-domain expansions.
In particular, the complex spherical harmonics of degree $\ell\ge 0$ and order $-\ell\le u\le\ell$ are defined as \cite{Costa2010Unified}:
\begin{equation}
Y_{\ell,u}(\theta^{\text{el}}, \theta^{\text{az}})
= (-1)^u\,N_{\ell, u}\,P_{\ell,u}(\cos\theta^{\text{el}})\,e^{\,j u\theta^{\text{az}}},
\label{eq:Ylm_complex}
\end{equation}
where the normalization constant is
$
N_{\ell, u}
=\sqrt{\frac{2\ell+1}{4\pi}\,\frac{(\ell-u)!}{(\ell+u)!}},
$
and $P_{\ell,u}(x)$ denotes the associated Legendre function.
These functions form an orthonormal set.
It is important to note that a finite-sized antenna element cannot produce any arbitrary combination of coefficients of the infinite set, but that the size of the elements leads to an upper limit on the effective number of modes that can be excited \cite{Han2021Characteristic,bucci2003spatial}. 
Thus, we use a truncated dictionary of size $Q$, by selecting the first $Q$ degree-order pairs ordered by increasing degree $\ell$ and, for equal $\ell$, by increasing order $u$ (e.g., $(0,0),(1,-1),(1,0),(1,1),(2,-2),\dots$), and denote the $k$-th basis by $Y_k(\theta^{\text{el}},\theta^{\text{az}})=b_k(\boldsymbol{\theta})$ \cite{Ying2024Reconfigurable}.

To adhere to the law of energy conservation, a normalization constraint is imposed such that $\lVert \mathbf{e}_{m,t}\rVert^2=1$ \cite{Ying2024Reconfigurable, Fadakar2026Hybrid}. 
Subsequently, the coefficient matrix $\mathbf{E}_t\!\in\!\mathbb{C}^{M\times MQ}$ is defined as:
\begin{equation}\label{def:E_synthesis}
\mathbf{E}_t = \mathrm{blkdiag}\bigl\{\mathbf{e}_{1,t}^{\mathsf{H}},\,\dots,\,\mathbf{e}_{M,t}^{\mathsf{H}}\bigr\}.
\end{equation}
Then, $\mathbf{q}_{t}(\mathbf{p}_u)$ in \eqref{eq:q-def} can be represented as:
\begin{equation}\label{eq:q_represent}
\mathbf{q}_{t}(\mathbf{p})
= \mathbf{E}_t\,\mathbf{d}(\mathbf{p}),
\end{equation}
where $\mathbf{d}(\mathbf{p})\in\mathbb{C}^{MQ}$ 
denotes the effective \ac{arv} defined as
$ 
[\mathbf{d}(\mathbf{p})]_{(m-1)Q+1:mQ}
= 
[\mathbf{a}(\mathbf{p})]_m
\mathbf{b}(\boldsymbol{\theta}_m)
$.
By substituting \eqref{eq:q_represent} into \eqref{eq:ch-def}, the received signal expression in \eqref{eq:y-def} is obtained as:
\begin{equation}\label{eq:y-def-simp}
y_{t}
=
\sqrt{P}
\left(
\beta_0
\mathbf{d}(\mathbf{p}_u)^{\mathsf{T}}
+
\sum_{i=1}^{I}
\beta_i
\mathbf{d}(\mathbf{p}_s^{(i)})^{\mathsf{T}}
\right)
\mathbf{E}_t^{\mathsf{T}}\,
\mathbf{f}_{t}
+
v_{t}\,.
\end{equation}

\section{Proposed Hybrid Design}\label{sec:hybrid_design}
This section derives the \ac{crb} for the 3D \ac{ue} position, which will then serve as the optimization objective for the joint design of the digital and \ac{em} precoders. 
Consistent with the approach in \cite{fadakar2024multi, fascista2022ris, fadakar2025mutual}, \ac{nlos} paths are treated as interference due to their significant path loss and the unknown number of such components. 
Consequently, we define the multipath- and noise-free version of \eqref{eq:y-def} as $x_t=\sqrt{P}\,\beta\,\mathbf{q}_{t}(\mathbf{p}_u)^\mathsf{T} \mathbf{f}_{t}=\sqrt{P}\,\beta\,\mathbf{d}(\mathbf{p}_u)^{\mathsf{T}}\mathbf{E}_t^\mathsf{T}\mathbf{f}_{t}$ where $\beta=\beta_0$ is defined in \eqref{eq:beta}. 
Let $\eta=[\mathbf{p}_u^\mathsf{T},\rho,\varphi]$ denote the unknown state parameters, where $\rho$ and $\varphi$ are the magnitude and phase components of $\beta$. 
The $(i,j)$-th element of the resulting \ac{fim} $\mathbf{J}\in\mathbb{C}^{5\times 5}$ is then expressed as \cite{fadakar2025near,fascista2022ris,Zhang2025Joint}:
\begin{align}\label{eq:fim_def}
[\mathbf{J}]_{i,j}
& =
\frac{2}{\sigma^2_v}
\sum_{t=1}^{N_t}
\Re
\bigg\{
\left(
\frac{\partial x_t}{\partial [\boldsymbol{\eta}]_{i}}
\right)^\mathsf{H}
\left(
\frac{\partial x_t}{\partial [\boldsymbol{\eta}]_{j}}
\right)
\bigg\}
,
\end{align}
Thus, the \ac{peb} at the \ac{ue} position $\mathbf{p}_\text{u}$ is given by:
\begin{equation}\label{eq:peb_def}
\text{PEB}
=
\mathrm{tr}([\mathbf{J}_{\boldsymbol{\eta}}^{-1}]_{1:3,1:3}).
\end{equation}
\subsection{Codewords Under Perfect Knowledge of UE Position}
Under the assumption of perfect knowledge of the \ac{ue} position, we formulate a joint optimization problem to minimize the \ac{peb} by designing the \ac{em} precoders $\{\mathbf{E}_t\}_{t=1}^{N_t}$ and digital precoders $\{\mathbf{f}_t\}_{t=1}^{N_t}$. The problem is stated as follows:
\begin{subequations} \label{optimization_problem}
\begin{align}
\underset{
\{\mathbf{E}_t,\mathbf{f}_t\}_{t=1}^{N_t}
}
{\textnormal{min}}
\quad  
&
\mathrm{PEB}
\left(
\{\mathbf{E}_t\}_{t=1}^{N_t},
\{\mathbf{f}_t\}_{t=1}^{N_t};\ 
\mathbf{p}_u
\right)
\label{def:opt-prob}
\\
\textnormal{s.t.} \quad
&
\lVert [\mathbf{E}_t]_{m,[(m-1)Q+1:mQ]}\rVert^2=1, \label{opt:em_const}
\\
&
\sum_{t=1}^{N_t}\mathbf{f}_t^\mathsf{H}\mathbf{f}_t=1, \label{opt:pow_const}
\\
&
m=1,\dots, M,\ 
t=1,\dots, N_t, \notag
\label{opt-const1}
\end{align}
\end{subequations}
where the constraints are defined such that \eqref{opt:em_const} enforces the unit normalization, while \eqref{opt:pow_const} ensures the total transmitted power $P$ is constant across all transmissions. 
Nonetheless, due to the coupling between the \ac{em} and digital precoders, the problem in \eqref{optimization_problem} is high-dimensional. 
To make it more tractable, we define the composite vectors $\{\mathbf{w}_t\}_{t=1}^{N_t}$ as follows:
\begin{equation}\label{def:w}
\mathbf{w}_t
=
\mathbf{E}_t^\mathsf{T}
\mathbf{f}_{t}
\in 
\mathbb{C}^{MQ}.
\end{equation}

Using this definition, our approach to solving \eqref{optimization_problem} involves a two-step process: 
we first find the \emph{low-dimensional structure} of the vectors $\{\mathbf{w}_t\}_{t=1}^{N_t}$, and then we determine the optimal structures for the \ac{em} and digital precoders. 
To this end, we first state the following lemma:
\begin{lemma}\label{lemma:W_affine}
The \ac{peb} in \eqref{eq:peb_def} is a convex function of $\mathbf{W}=\sum_{t=1}^{N_t}\mathbf{W}_t$ with $\mathbf{W}_t=\mathbf{w}_t\mathbf{w}_t^\mathsf{H}$.
\end{lemma}
\begin{proof}
Note that $x_t=\sqrt{P}\,\beta\,\mathbf{d}(\mathbf{p}_u)^{\mathsf{T}}\mathbf{E}_t^\mathsf{T}\mathbf{f}_{t}=\sqrt{P}\,\beta\,\mathbf{d}(\mathbf{p}_u)^{\mathsf{T}}\mathbf{w}_{t}$. 
Thus, for instance, the elements of the first $3\times 3$ submatrix can be obtained as:
\begin{align}\label{eq:fim_rep}
[\mathbf{J}]_{i,j}
& =
\frac{2P}{\sigma^2_v}
\sum_{t=1}^{N_t}
\Re
\bigg\{
\mathbf{d}_{t}^{(i)}(\mathbf{p}_u)^\mathsf{T}
\mathbf{w}_{t}
\mathbf{w}_{t}^\mathsf{H}
\mathbf{d}^{(j)}_{t}(\mathbf{p}_u)^{*}
\bigg\} \notag \\
& =
\frac{2P}{\sigma^2_v}
\Re
\bigg\{
\mathbf{d}_{t}^{(i)}(\mathbf{p}_u)^\mathsf{T}
\mathbf{W}
\mathbf{d}^{(j)}_{t}(\mathbf{p}_u)^{*}
\bigg\}
,
\end{align}
where 
$
\mathbf{d}^{(i)}(\mathbf{p})
=
\partial\mathbf{d}(\mathbf{p})/
\partial [\mathbf{p}_u]_{i}
$. 
Hence, these elements depend linearly on the matrix $\mathbf{W}$, and in a similar manner, the remaining \ac{fim} elements can also be shown to exhibit linear dependence on $\mathbf{W}$ (details omitted for space efficiency). 
As a result, the \ac{peb} in \eqref{eq:peb_def} is a convex function of $\mathbf{W}$ via the composition rule \cite{boyd2004convex}.
\end{proof}

To transform the non-convex problem in \eqref{optimization_problem} into a tractable convex form, we first express it in terms of the matrix $\mathbf{W}$. 
In this reformulation, we assume that a set of feasible digital and \ac{em} precoders $\{\mathbf{E}_t,\mathbf{f}_t\}_{t=1}^{N_t}$ can be recovered from the optimal solution for $\mathbf{W}$. 
Consequently, we can drop the constraint \eqref{opt:em_const}, leading to the following problem:
\begin{subequations} \label{opt_prob_represent}
\begin{align}
\underset{
\mathbf{W}
}
{\textnormal{min}}
\quad  
&
\mathrm{PEB}
\left(
\mathbf{W};
\mathbf{p}_u
\right)
\label{def:opt-prob-represent}
\\
\textnormal{s.t.} \quad
&
\mathrm{tr}\left(
\mathbf{W}
\right)=1
, \label{opt:pow_const_w}
\\
& 
\mathrm{rank}(\mathbf{W})\le N_t \,,
\label{opt:rank_1_constraints}
\end{align}
\end{subequations}
where the power constraint in \eqref{opt:pow_const_w} is equivalent to \eqref{opt:pow_const}. 
We then relax the problem by dropping \eqref{opt:rank_1_constraints} to obtain a convex problem. 
To achieve a low-complexity solution, in a similar manner to \cite[Appendix~B]{Fadakar2026Hybrid}, we exploit the inherent structure of the optimal matrix $\mathbf{W}$, which can be represented as:
\begin{equation}\label{eq:w_low_dim_structure}
\mathbf{W}
=
\mathbf{D}_w
\mathbf{\Xi}
\mathbf{D}_w^\mathsf{H} \,,
\end{equation}
where $\mathbf{\Xi}\in\mathbb{C}^{4\times 4}$ is a \ac{psd} matrix, and
$
\mathbf{D}_w
=
[
\mathbf{d}^{(0)}(\mathbf{p}_u)^{*},
\mathbf{d}^{(1)}(\mathbf{p}_u)^{*},
\mathbf{d}^{(2)}(\mathbf{p}_u)^{*},
\mathbf{d}^{(3)}(\mathbf{p}_u)^{*}
]\in \mathbb{C}^{MQ\times 4}
$, where $\mathbf{d}^{(0)}(\mathbf{p}_u)=\mathbf{d}(\mathbf{p}_u)$. 
Finally, by relaxing $\mathbf{\Xi}$ to be a diagonal matrix, we obtain the following four codewords $\{\mathbf{w}_i\}_{i=1}^{4}$ achieving the approximate optimal value of \eqref{opt_prob_represent}:
\begin{equation}\label{eq:w_opt_synthesis}
\hat{\mathbf{w}}_i=
\sqrt{\varrho_i}\,
\tilde{\mathbf{d}}^{(i)}(\mathbf{p}_u)^{*},
\end{equation}
where 
$\tilde{\mathbf{d}}^{(i)}(\mathbf{p}_u)
=
\mathbf{d}^{(i)}(\mathbf{p}_u)/\lVert \mathbf{d}^{(i)}(\mathbf{p}_u)\rVert$
and $1\ge \varrho_i\ge 0$ is the portion of power dedicated to the $i$-th codeword.

\begin{proposition}\label{prop:opt_precoders}
Given the four codewords obtained in \eqref{eq:w_opt_synthesis}, the associated digital and \ac{em} precoders are given by:
\begin{align}\label{optimal_EM_f_synthesis}
[\hat{\mathbf{f}}_{i}]_m
=
\sqrt{\varrho_i}
\lVert\mathbf{d}_{m}^{(i)}(\mathbf{p}_u)\rVert
e^{j \psi_{m,i}}
,\ 
\hat{\mathbf{e}}_{m,i}
=
\frac{\mathbf{d}_{m}^{(i)}(\mathbf{p}_u)^{*}}{\lVert \mathbf{d}_{m}^{(i)}(\mathbf{p}_u)\rVert}
e^{-j \psi_{m,i}},
\end{align}
where $\mathbf{d}^{(i)}_m(\mathbf{p}_u)=[\tilde{\mathbf{d}}^{(i)}(\mathbf{p}_u)]_{((m-1)Q+1):mQ}$, and $\{\psi_{m,i}\}_{m=1,i=0}^{M,3}$ are arbitrary phases.
\end{proposition}
\begin{proof}
To obtain the corresponding digital and \ac{em} precoders, we use \eqref{def:w}, which establishes a relation between the vector $\mathbf{w}_i$, $\mathbf{E}_i$ and $\mathbf{f}_{i}$. 
According to the definition of $\mathbf{E}_i$, we obtain the following equation:
\begin{equation}\label{eq:w_f_e}
[\mathbf{w}_i]_{(m-1)Q+1:mQ}
=
[\mathbf{f}_{i}]_m
\mathbf{e}_{m,i}
\end{equation}
for $m=1,\dots ,M$. 
By substituting \eqref{eq:w_opt_synthesis} in \eqref{eq:w_f_e}, we obtain:
$
[\mathbf{f}_{i}]_m
\mathbf{e}_{m,i}
=
\sqrt{\varrho_i}\,\mathbf{d}^{(i)}_m(\mathbf{p}_u)^{*}.
$
Since $\lVert \mathbf{e}_{m,i}\rVert^2 = 1$, the unique solutions in \eqref{optimal_EM_f_synthesis} are found for $[\mathbf{f}_{i}]_m$ and $\mathbf{e}_{m,i}$ which completes the proof.
\end{proof}
\vspace{-0.5cm}
\subsection{Proposed Codebook Under UE Position Uncertainty}\label{sec:codebook_synthesis}
In the previous section, we derived a low-complexity, joint digital and \ac{em} codebook from the optimal solutions of \eqref{optimization_problem}, designed to minimize the \ac{peb} under the assumption of exact state knowledge of the \ac{ue} position. 
In practice, however, the \ac{ue} location is typically uncertain, as there would be no need for a measurement if we knew it exactly.
Let $\bm{\mathcal{U}}$ denote the 3D region that captures the \ac{ue} position uncertainty. 
To design for robustness, we sample $L$ uniformly spaced candidate \ac{ue} locations $\{\mathbf{p}_u^{(\ell)}\}_{\ell=1}^{L}$ that cover $\bm{\mathcal{U}}$. 
We then construct the following codebook of $N_t=4L$ codewords $\widehat{\mathbf{w}}_i$ (from \eqref{eq:w_opt_synthesis}):
\begin{align}
\bm{\mathcal{W}}
=
\bigg\{
&\sqrt{\varrho_{4\ell -3}} 
\widehat{\mathbf{w}}_0(\mathbf{p}_u^{(\ell)}),
\sqrt{\varrho_{4\ell -2}} 
\widehat{\mathbf{w}}_1(\mathbf{p}_u^{(\ell)}),\notag \\
&
\sqrt{\varrho_{4\ell-1}} 
\widehat{\mathbf{w}}_3(\mathbf{p}_u^{(\ell)})
,
\sqrt{\varrho_{4\ell}} 
\widehat{\mathbf{w}}_4(\mathbf{p}_u^{(\ell)})
\bigg\}_{\ell=1}^{L}.
\label{def:synthesis_codebook_w}
\end{align}
The corresponding digital and \ac{em} precoders are obtained from Prop.~\ref{prop:opt_precoders}. 
In \eqref{def:synthesis_codebook_w}, we explicitly separate the power allocation coefficients to highlight that they are the only remaining design variables and to streamline the presentation of their optimization in the next subsection.

The optimization problem for power allocation within the uncertainty region $\bm{\mathcal{U}}$ is formulated as follows:
\begin{subequations} \label{opt:power_alloc}
\begin{align}
\underset{
\{\varrho_t\}_{t=1}^{N_t}
}
{\textnormal{min}}
\,
\underset{
\bm{p}_{u}\in\bm{\mathcal{U}}
}
{\textnormal{max}}
\quad  
&
\mathrm{PEB}
\left(
\{
\varrho_t\widehat{\mathbf{W}}_t
\}_{t=1}^{N_t}
;\ 
\mathbf{p}_{u}
\right)
\\
\textnormal{s.t.} \quad
&
\varrho_t\ge 0,\ 
\sum_{t=1}^{N_t}
\varrho_t=1,
\label{opt:delta_sum}
\end{align}
\end{subequations}
where $\varrho_t\widehat{\mathbf{W}}_t=\varrho_t\widehat{\mathbf{w}}_t\widehat{\mathbf{w}}_t^\mathsf{H}$ denotes the covariance matrix associated with the $t$-th codeword in \eqref{def:synthesis_codebook_w}. 
By utilizing the epigraph form, we first derive the semi-infinite version of this optimization problem \cite{boyd2004convex}:
\begin{subequations} \label{opt:power_alloc_step2}
\begin{align}
\underset{
s,\{\varrho_t\}_{t=1}^{N_t}
}
{\textnormal{min}}
\quad  
&
s
\\
\textnormal{s.t.} \quad
&
\mathrm{PEB}
\left(
\{
\varrho_t\widehat{\mathbf{W}}_t
\}_{t=1}^{N_t}
;\ 
\mathbf{p}_{u}
\right)
\le s,\ 
\forall \mathbf{p}_u\in\bm{\mathcal{U}},
\\
& 
\eqref{opt:delta_sum}
.
\end{align}
\end{subequations}
By discretizing the uncertainty region $\bm{\mathcal{U}}$ into a finite set of $N_u$ points $\{\mathbf{p}_{u,i}\}_{i=1}^{N_u}$, we convert the semi-infinite problem into the following approximated formulation using the Schur complement \cite{boyd2004convex}:
\begin{subequations} \label{opt:power_alloc_step3}
\begin{align}
& 
\underset{
s,
\{\varrho_t\}_{t=1}^{N_t}, 
\{u_{i,m}\}_{i=1,m=1}^{N_u,3}
}
{\textnormal{min}}
s
\\
\textnormal{s.t.} \quad
&
\begin{bmatrix}
\mathbf{J}_{\bm{\eta}}\left(
\{
\varrho_t\widehat{\mathbf{W}}_t
\}_{t=1}^{N_t}
;\ 
\mathbf{p}_{u,i}
\right) & [\mathbf{I}_5]_{:,m}\\
[\mathbf{I}_5]_{:,m}^{\mathsf{T}} & u_{i,m}
\end{bmatrix}
\succeq 0,
\label{opt:lmi_ineq}
\\
& 
\sum_{m=1}^{3}u_{i,m} \le s,
\\
& 
\eqref{opt:delta_sum}, \notag
\\
&
m=1,2,3,\ i=1,\dots, N_u.
\notag
\end{align}
\end{subequations}
Based on the result of Lemma~\ref{lemma:W_affine}, the elements of the \ac{fim} are affine with respect to the terms $\{\varrho_{t}\widehat{\mathbf{W}}_{t}\}_{t=1}^{N_{t}}$. 
This confirms that the constraint in \eqref{opt:lmi_ineq} is a \ac{lmi} in the optimization variables $\{\varrho_{t}\}_{t=1}^{N_{t}}$ and $\{u_{i,m}\}$. 
Therefore, the power-allocation problem \eqref{opt:power_alloc_step3} is a convex \ac{sdp} \cite{boyd2004convex} in the coefficients $\{\varrho_{t}\}_{t=1}^{N_t}$, and can thus be solved efficiently using readily available convex optimization solvers.
\vspace{-0.2cm}
\section{Proposed Localization Approach}\label{sec:ML_localization}
The \ac{ml} localization problem can be formulated as follows:
\begin{equation}\label{eq:localization_formulation}
\hat{\mathbf{p}}_u
=
\argmin_{\mathbf{p}}
\lVert
\mathbf{y}
-
\zeta
\mathbf{x}(\mathbf{p})
\rVert^2,
\end{equation}
where $\mathbf{y}\in\mathbb{C}^{N_t}$ is the received signal vector after stacking all $N_t$ signals in \eqref{eq:y-def}, and $\zeta=\sqrt{P}\,\beta$, $\mathbf{x}(\mathbf{p})^\mathsf{T}=\mathbf{d}(\mathbf{p})^{\mathsf{T}}\mathbf{W}$ are the likelihood signal models, where $\mathbf{W}=[\mathbf{w}_1,\dots ,\mathbf{w}_{N_t}]$.  
$\zeta$ 
can be estimated via \ac{ls} technique as 
$
\hat{\zeta}(\mathbf{p}) 
=
\frac{
\mathbf{x}(\mathbf{p})^{\mathsf{H}}
\mathbf{y}}{\lVert
\mathbf{x}(\mathbf{p})
\Vert^2}\,
$. 
After substituting this into \eqref{eq:localization_formulation} and simplifying, we obtain:
\begin{equation}\label{eq:opt_simp}
\hat{\mathbf{p}}_u
=
\argmax_{\mathbf{p}}
\frac{
\lvert
\mathbf{x}(\mathbf{p})^{\mathsf{H}}
\mathbf{y}
\rvert^2}{\lVert
\mathbf{x}(\mathbf{p})
\Vert^2}.
\end{equation}
The proposed method involves two stages, i.e., coarse positioning and position refinement, which are explained next.

\subsection{Coarse Stage}\label{sec:loc_coarse}
The uncertainty region is discretized into a grid mesh with coarse step sizes $x_g$, $y_g$ and $z_g$ along the respective dimensions. 
For each grid point $\bm{p}_{u,g}$, we compute the normalized vector
$
\frac{\mathbf{x}(\mathbf{p})}{\lVert \mathbf{x}(\mathbf{p}) \rVert}
$. 
By stacking these normalized vectors, we form the matrix $\mathbf{X} \in \mathbb{C}^{N_t \times N_g}$. 
To estimate the \ac{ue} position for a given observation vector $\bm{y}$, we first calculate the vector $\mathbf{z}=\mathbf{X}^\mathsf{H} \mathbf{y}$. 
The corresponding position estimate is then determined by selecting the grid point with the maximum absolute value in $\mathbf{z}$.
To improve efficiency, the matrix $\mathbf{X}$ can be precomputed and reused for positioning at any location.
\subsection{Refinement Stage}\label{sec:refine_stage}
The coarse-stage solution uses the accurate problem formulation \eqref{eq:opt_simp}. 
However, the grid-based approach only evaluates the solution at discrete points. 
In high \ac{snr} regimes, where the true optimal position lies between these grid points, this \emph{quantization effect} results in unavoidable numerical inaccuracies.
To mitigate this issue, we employ a gridless quasi-Newton approach \cite{fadakar2024multi}. 
This method uses the initial estimate obtained from the previous subsection to perform a precise, refined optimization of the position.
\vspace{-0.2cm}
\section{Simulations}
This section presents numerical simulations that evaluate the proposed joint \ac{em} and digital precoder design and its impact on localization in \ac{ra}-assisted mmWave systems. 
The default parameter set is given in Table~\ref{tab:sys-params}. 
Some parameters may vary as specified in each experiment. 
In the coarse localization stage (Sec.~\ref{sec:loc_coarse}), we first obtain an initial estimate using a step size of $x_g=y_g=z_g=1\,\mathrm{m}$. 
This estimate is subsequently refined by reapplying the same stage with a reduced step size of $x_g=y_g=z_g=0.1\,\mathrm{m}$ within a unit cube neighborhood. 
The final, accurate, gridless positioning is then achieved by utilizing the dedicated refinement stage explained in Sec.~\ref{sec:refine_stage}.

\begin{table}[ht]
\caption{\label{tab:sys-params} System parameters}
\centering
\fontsize{12}{10}\selectfont 
\resizebox{\columnwidth}{!}{
\begin{tabular}{ |l|l|  }
\hline
Default Parameters
&
\textbf{Value}
\\
\hline
Carrier frequency $f_c$ & $30\,\mathrm{GHz}$
\\
Noise PSD $N_0$ & $-173.855\,\mathrm{dBm}$
\\
Light speed $c$ & $3\times 10^8\,\mathrm{m/s}$
\\
Bandwidth $B$ & $1\,\mathrm{MHz}$
\\
UE position $\mathbf{p}_u$ & $[10.35, 1.67, 0]^\mathsf{T}$
\\
Number of bases 
$Q$ &  $20$
\\
Uncertainty region $\mathbf{\mathcal{U}}$ & $9<x<14,-3<y<3$, $0<z<4$
\\
\ac{bs} position $\mathbf{p}_b$ & $[0,0,5]^\mathsf{T}\,\mathrm{[m]}$
\\
\ac{bs} number of elements $M^{h}$, $M^v$ & $50$, $50$
\\
Localization parameters $N_\tau$, $N_\theta$ & $1000$, $500$
\\
\hline
\end{tabular}
}
\end{table}

\vspace{-0.5cm}
\subsection{Localization Performance Versus SNR}
Fig.~\ref{fig:RMSE_stages_SNR} compares the \ac{rmse} of the proposed coarse and fine localization stages with the theoretical \ac{peb}. 
The fine stage approaches the \ac{peb} (around $-5\,\mathrm{dB}$), while the coarse stage, due to its discrete nature, exhibits a performance floor and saturates at sufficiently high \ac{snr}.

Fig.~\ref{fig:RMSE_CRB_vs_SNR} plots the \ac{rmse} together with the corresponding \ac{peb} for the \ac{ml}-based estimator in Sec.~\ref{sec:ML_localization}, and benchmarks our \ac{ra} with joint \ac{em} and digital precoder optimization against the following baselines: 
the same \ac{ra} equipped with heuristic directional \ac{em} and digital precoders (i.e., using only directional codewords in \eqref{def:synthesis_codebook_w} with indices $\{4\ell-3\}_{\ell}^{L}$), and a conventional non-reconfigurable array using the optimal codebooks from \cite{Keskin2022Optimal, fascista2022ris} while matching aperture and antenna count for fairness. 
To establish a fair comparison, optimal power allocation is applied to the codewords for every method evaluated.
The results show that the \ac{ra} combined with our optimized precoders achieves substantially better localization accuracy than both the traditional array and heuristic directional beamforming, and that the \ac{ra} approach is notably more robust in low-\ac{snr} regimes.

\begin{figure}[!t]
\centering
\includegraphics[width=\columnwidth]{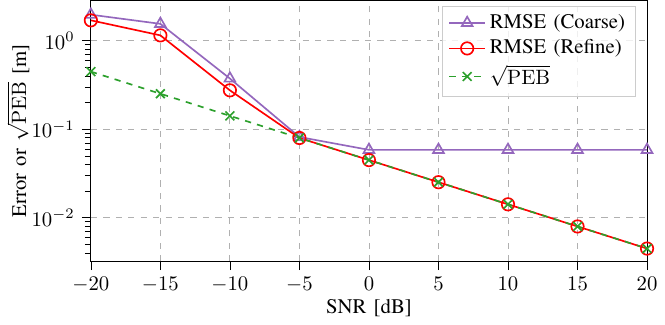}
\caption{
Localization performance of the proposed two-stage localization algorithm
}
\label{fig:RMSE_stages_SNR}
\end{figure}

\begin{figure}[!t]
\centering
\includegraphics[width=\columnwidth]{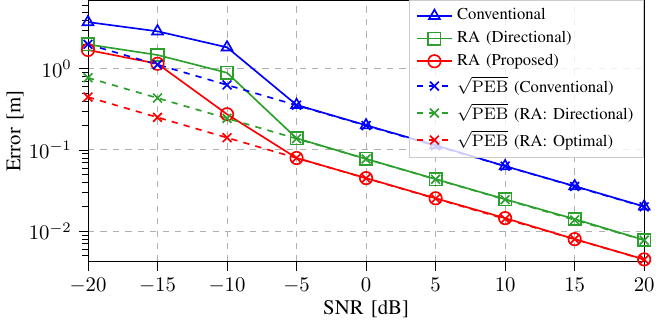}
\caption{
Localization performance comparison with different baselines.
}
\label{fig:RMSE_CRB_vs_SNR}
\end{figure}
\subsection{Performance Under Different Number of Bases}
Fig.~\ref{fig:CRB_vs_Q} shows the \ac{peb} as a function of the number of \ac{shod} bases $Q$, using an optimized non-reconfigurable array as baseline. 
The proposed method yields steadily improving \ac{peb} as $Q$ increases, whereas the baseline remains unchanged since it does not depend on $Q$. 
Notably, the two curves coincide at $Q=1$, which stems from the fact that the first \ac{shod} basis is constant on the unit sphere, so in this case, the \ac{ra} array reduces to the conventional optimal array.

\begin{figure}[!t]
\centering
\includegraphics[width=\columnwidth]{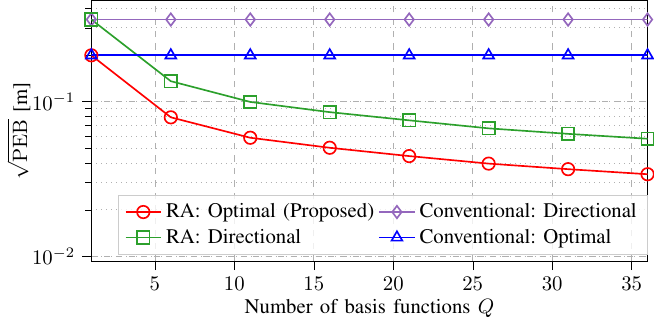}
\caption{
Localization performance versus different numbers of basis functions.
}
\label{fig:CRB_vs_Q}
\end{figure}

\subsection{Impact of Interference}\label{sec:simul_interference}
In this subsection, we examine the robustness of the proposed joint \ac{em} and digital precoder design to interference. 
We benchmark the proposed \ac{ra} with optimized joint precoders against the same \ac{ra} using heuristic directional beamforming and a conventional non-reconfigurable array. 
Interference is modeled by $I=10$ \ac{mpc_p} whose locations are drawn uniformly inside the \ac{ue} uncertainty region $\bm{\mathcal{U}}$ to create a challenging scenario. 
Fig.~\ref{fig:RMSE_CRB_vs_LMR} reports \ac{rmse} and \ac{peb} as functions of \ac{lmr} (using \cite[Eq.~(24)]{fadakar2024multi}) at $\mathrm{SNR}=15\,\mathrm{dB}$; each \ac{rmse} point is obtained from 500 Monte-Carlo trials. 
Because the \ac{fim} in Sec.~\ref{sec:hybrid_design} accounts only for the \ac{los} path, the \ac{peb} remains constant for each \ac{lmr}. The \ac{rmse} approaches the \ac{peb} near $\mathrm{LMR}=5\,\mathrm{dB}$, and the proposed \ac{ra} with optimized precoders is close to the \ac{peb} for $\mathrm{LMR}\in[10\,\mathrm{dB},15\,\mathrm{dB}]$. 
Performance degrades at low \ac{lmr} because strong multipath and a weak \ac{los} component cause model mismatch and hence poorer localization. 

\begin{figure}[!t]
\centering
\includegraphics[width=\columnwidth]{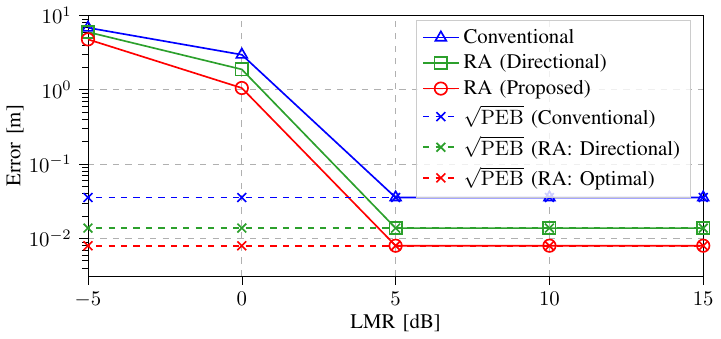}
\caption{
Localization performance under interference.
}
\label{fig:RMSE_CRB_vs_LMR}
\end{figure}

\section{Conclusion}
In this work, we addressed the joint design of digital and \ac{em} precoders for maximizing downlink localization accuracy in \ac{nf} scenarios utilizing an \ac{ra} array. 
Our approach adopted the synthesis reconfigurability model, wherein individual antenna elements synthesize beampatterns from a set of orthonormal basis functions. 
We developed low-complexity beamforming designs for both the digital and \ac{em} beamforming. 
Furthermore, we proposed an efficient \ac{ml}-based localization algorithm. 
Extensive simulations validated our approach, demonstrating substantial improvements in the \ac{peb} and \ac{rmse} compared to conventional non-reconfigurable arrays and purely directional beamforming schemes. 
Future research will extend these methods to \ac{isac} systems and investigate their robustness against model mismatch and practical hardware impairments.

\section*{Acknowledgment}
This work is supported, in part, 
by the National Science Foundation (Grants 2229535 and 2106602), 
in part by the KAUST Global Fellowship Program under Award No. RFS-2025-6844, and
in part by the Swedish Research Council (Grants 2022-03007 and 2024-04390).
\bibliographystyle{IEEEtran}
\bibliography{Bib}

\end{document}